\documentclass[12pt]{amsart}
\usepackage[english]{babel}
\usepackage[utf8]{inputenc}
\usepackage{times}
\usepackage[T1]{fontenc}
\usepackage{graphicx}
\usepackage{amscd,amsmath,amsfonts,amstext,amssymb,amsbsy,amsopn,amsthm,eucal}
\usepackage{txfonts}

\usepackage{fancyhdr}

\usepackage{hyperref}

\hypersetup{
  pdftitle={Polyharmonic Laplacian},
  pdfauthor={Daguang Chen},
  pdfsubject={The sharp estimates of eigenvalues of Polyharmonic Laplacian},
  pdfkeywords={eigenvalue, Polyharmonic operator, higher order Stokes operator}
  pdfpagelayout=SinglePage,
  pdfpagemode=UseOutlines,  pdfpagemode=UseOutlines,
  colorlinks,
  linkcolor=[rgb]{0,0,0.7},
  urlcolor=[rgb]{0,0,0.4},
  citecolor=[rgb]{0.4,0.1,0}
  }

\newcommand{\n}{\nabla}

\renewcommand{\div}{\mathop\mathrm{div}}

\newcommand{\ld}{\lambda}
\renewcommand{\u}{\textbf{u}}
\renewcommand{\v}{\textbf{v}}

\newtheorem{prop}{Proposition}[section]
\newtheorem{thm}[prop]{Theorem}

\newtheorem{lem}[prop]{Lemma}

\newtheorem{rem}[prop]{Remark}
\newtheorem{cor}[prop]{Corollary}

\numberwithin{equation}{section}

\begin{document}


\baselineskip=17pt


\title[Eigenvalues of Polyharmonic Laplacian and Higher Order Stokes Operator]{The sharp estimates of eigenvalues of Polyharmonic operator and higher order Stokes operator}

\author[Daguang Chen  and  He-Jun Sun]{Daguang Chen  and He-Jun Sun}


\maketitle


\begin{abstract}
In this paper, we establish some lower bounds for the sums of eigenvalues of the polyharmonic operator and  higher order Stokes operator,
which are sharper than the recent results in \cite{CSWZ13, I13}.
At the same time, we obtain some certain bounds for the sums of positive and negative powers of eigenvalues of the polyharmonic operator.
\end{abstract}

\section{Introduction}

Let $\Omega$ be a bounded domain in an $n$-dimensional
Euclidean space $\mathbb{R}^{n}$ ($n\geq 2$).
The Dirichlet eigenvalue problem of the polyharmonic operator is described by
\begin{equation}\label{PL}
\begin{cases}(-\Delta)^{l}u = \lambda u,\quad \quad \mbox{on}  \ \Omega,\\
u|_{\partial \Omega}=\frac{\partial u}{\partial\nu}|_{\partial \Omega}= \cdots = \frac{\partial^{l-1} u}{\partial\nu^{l-1}}|_{\partial \Omega} =0,
\end{cases}\end{equation}
where $\Delta$ is the Laplacian and $\nu$ denotes the outward unit normal vector field of $\partial\Omega$.
As we known, this problem has a real and discrete spectrum:
$$
0 < \lambda_1 \leq \lambda_2 \leq \cdots \leq \lambda_k \leq \cdots \rightarrow \infty,
$$
where each eigenvalue repeats with its multiplicity.

When $l=1$, problem (\ref{PL}) is called the Dirichlet Laplacian problem or the fixed membrane problem. The asymptotic behavior of its $k$-th eigenvalue $\lambda_k$ relates to geometric properties of $\Omega$ when $k \rightarrow \infty$.
In fact, the following Weyl's asymptotic formula asserts that
\begin{equation}\label{Weyl}
\lambda_k  \sim
   C_n\left(\frac {k}{|\Omega|}\right)^{\frac{2}{n}}, \quad \mbox{as} \   k \rightarrow \infty,
\end{equation}
where $C_n=4\pi\Gamma(1+\frac n2)^{\frac 2n}$ and $|\Omega|$ denote the volume of $\Omega$.
Here $\Gamma(m)$ denotes the Gamma function $\Gamma(m) = \int_0^{\infty} t^{m-1} e^{-t} dt$ for $m > 0$.
In 1961, P\'{o}lya proved in \cite{P} that
\begin{equation}\label{Polya1}
\lambda_k  \geq
   C_n\left(\frac {k}{|\Omega|}\right)^{\frac{2}{n}}
\end{equation}
for tilling domain in $\mathbb{R}^{n}$. Moreover, he conjectured that (\ref{Polya1}) holds for any bounded domain in $\mathbb{R}^n$.
There have been some results in this direction. In 1980, Lieb \cite{Lie} proved
\begin{equation*}
\lambda_k  \geq
   \tilde C_n\left(\frac {k}{|\Omega|}\right)^{\frac{2}{n}},
\end{equation*}
where $\tilde C_n$  differs from $C_n$ in (\ref{Polya1}) by a factor.
In 1983, Li and Yau \cite{LY} proved
\begin{equation}\label{BLY}
 \sum_{j=1}^k  \lambda_j  \geq
4\pi\frac{n}{n+2} \left(\frac{\Gamma(1+\frac n2)}{|\Omega|}\right)^{\frac{2}{n}}   k^{1+ \frac{2}{n}}.
\end{equation}
It has been pointed out in \cite{LW} that  by using the Legendre transform, (\ref{BLY})
is equivalent to the inequality  derived by Berezin \cite{Ber}.
Hence, (\ref{BLY}) is also called the Berezin-Li-Yau inequality.
Using the similar approach, Kr\"{o}ger \cite{Kroger} has obtained the sharp upper bound for the Neumann eigenvalues.
Improvements to the Berezin–Li–Yau inequality in (\ref{BLY}) for the case of Dirichlet
Laplacian have appeared recently (for example, see \cite{KVW,M,W}). In particular, Melas \cite{M} improved (\ref{BLY}) to
\begin{equation}\label{Melas}
  \sum_{j=1}^k  \lambda_j \geq
4\pi\frac{n}{n+2} \left(\frac{\Gamma(1+\frac n2)}{|\Omega|}\right)^{\frac{2}{n}} k^{1+\frac{2}{n}}
+ \frac{1}{24(n+2)} \frac{|\Omega|}{I(\Omega)}k,
\end{equation}
where  $I(\Omega) = \underset{a \in \mathbb{R}^n}{\textrm{min}}  \int_\Omega |x-a|^2 dx$ is the moment of inertia of $\Omega$ and $a$ is a constant vector in $\Bbb R^n$.
In 2010, Ilyin \cite{I10} obtained the following asymptotic lower bound for eigenvalues of this problem:
\begin{equation}
   \begin{aligned}
 \sum_{j=1}^k  \lambda_j
 \geq &  4\pi\frac{n}{n+2} \left(\frac{\Gamma(1+\frac n2)}{|\Omega|}\right)^{\frac{2}{n}} k^{1+\frac{2}{n}}
  +  \frac{n}{  48}     \frac{ |\Omega|  }{  I(\Omega)} k  \bigg(1- \varepsilon_n(k) \bigg),
  \end{aligned}
\end{equation}
where $0  \leq \varepsilon_n(k) = O( k^{-\frac{2}{n}}) $ is a infinitesimal of $k^{-\frac{2}{n}}$.
In 2013, Y. Yolcu and T. Yolcu \cite{YY2012} proved that
\begin{equation}\label{YY12}
 \begin{aligned}
 \sum_{j=1}^k  \lambda_j
 \geq &  4\pi\frac{n}{n+2} \left(\frac{\Gamma(1+\frac n2)}{|\Omega|}\right)^{\frac{2}{n}} k^{1+\frac{2}{n}}
  +\frac{2\sqrt{\pi}}{n+2}\left(\frac{|\Omega|}{I(\Omega)}\right)^{\frac{1}{2}}\left(\frac{\Gamma(1+\frac n2)}{|\Omega|}\right)^{\frac{1}{n}}\\
  &-\frac{5}{8(n+2)}\frac{|\Omega|}{I(\Omega)}k
  +\frac{1}{16\sqrt{\pi}(n+2)}\left(\frac{|\Omega|}{I(\Omega)}\right)^{\frac{3}{2}}
 \left(\frac{\Gamma(1+\frac n2)}{|\Omega|}\right)^{-\frac{1}{n}}k^{1-\frac 1n}.
 \end{aligned}
\end{equation}

When $l=2$, problem (\ref{PL}) is called the clamped plate problem.
For the developments of eigenvalues of the clamped plate problem,
we refer the readers to \cite{A,CQW2011, CW1,I10,I13,LP,Pl,YY2013}.

For any order $l$, Levine and Protter \cite{LP} proved
\begin{equation}
 \sum^{k}_{j=1}\lambda_{j} \geq  (4\pi)^l\frac{n}{n+2l}    \left(\frac{\Gamma(1+\frac n2)}{|\Omega|}\right)^{\frac{2l}{n}}  k^{1+\frac{2l}{n}}.
\end{equation}
Recently, Cheng, Sun, Wei and Zeng \cite{CSWZ13} proved (see also Theorem 3 in \cite{I13})
\begin{equation}\label{CSWZ}
 \begin{aligned}
  \sum_{j=1}^k   \lambda_j
 \geq &  (4\pi)^l\frac{n }{n+2l}  \left(\frac{\Gamma(1+\frac n2)}{|\Omega|}\right)^{\frac{2l}{n}}  k^{1+\frac{2l}{n}} \\
&  + (4\pi)^{l-1}\frac{nl }{48}   \left(\frac{\Gamma(1+\frac n2)}{|\Omega|}\right)^{\frac{2l-2}{n}} \frac{ |\Omega|}{I(\Omega)}  k^{ \frac{2l-2}{n}}
  \left(1- \varepsilon_n(k) \right),
  \end{aligned}
\end{equation}
where $0  \leq \varepsilon_n(k) = O( k^{-\frac{2}{n}}) $ is a infinitesimal of $k^{-\frac{2}{n}}$.

In this article, we obtain the following estimates for the sum of eigenvalues of proplem \ref{PL}.
\begin{thm}\label{PLThm1}For any bounded domain $\Omega\subset \Bbb R^n (n\geq 2), k\geq 1$ and $1\leq l< \frac {n+1}{2}$, the eigenvalues of (\ref{PL}) satisfy
  \begin{equation}\label{Mainineq2}
    \begin{aligned}
      \sum_{j=1}^k\ld_j\geq & \frac{n(4\pi)^l}{n+2l}\left(\frac{\Gamma(1+\frac n2)}{|\Omega|}\right)^{\frac{2l}{n}}k^{1+\frac{2l}{n}}
      +\frac{(4\pi)^{\frac{2l-1}{2}}l}{(n+2l)I(\Omega)^{\frac12}}
      |\Omega|^{\frac12-\frac{2l-1}{n}}\Gamma(1+\frac n2)^{\frac{2l-1}{n}}k^{1+\frac{2l-1}{n}}\\
      &-\frac{5(4\pi)^{l-1}l}{8(n+2l)I(\Omega)}\Gamma(1+\frac n2)^{\frac{2(l-1)}{n}}|\Omega|^{1-\frac{2l-2}{n}}k^{1+\frac{2l-2}{n}}\\
      &+\frac{(4\pi)^{\frac{2l-3}{2}}l}{8(n+2l)I(\Omega)^{\frac 32}}|\Omega|^{\frac 32-\frac{(2l-3)}{n}}\Gamma(1+\frac n2)^{\frac{(2l-3)}{n}}k^{1+\frac{2l-3}{n}}.
    \end{aligned}
  \end{equation}
\end{thm}
\begin{rem}
Theorem \eqref{PLThm1} is a generalization of \eqref{YY12} for problem \eqref{PL}.  The inequality \eqref{Mainineq2} give an improvement for \eqref{CSWZ}.
\end{rem}

Motivated by the work of \cite{I13, Vougalter, YY2013}, we  obtain the following estimates of the negative and positive power of eigenvalues of problem 1.1:
\begin{thm}\label{Vou0}
 For $0<q\leq 1$, $l \in\Bbb N$ and $2\leq n$, the sums of positive powers of eigenvalues of the polyharmonic Laplacian problem (\ref{PL})  on $\Omega$ satisfy
 \begin{equation}\label{V0}
  \begin{aligned}
  \sum_{j=1}^k\ld_j^q\geq & (4\pi)^{lq}\frac{n}{n+2lq}\left(\frac{\Gamma(1+n/2)}{|\Omega|}\right)^{\frac{2lq}{n}}
   k^{1+\frac{2lq}{n}}\\
   &+(4\pi)^{lq-1}\frac{nlq}{48}\frac{|\Omega|}{I(\Omega)}
   \left(\frac{\Gamma(1+n/2)}{|\Omega|}\right)^{\frac{2lq-2}{n}}k^{1+\frac{2lq-2}{n}}+O(k^{1+\frac{2lq-4}{n}}).   \end{aligned}
 \end{equation}
\end{thm}

\begin{cor}\label{Vou1}
For $0<q\leq 1$, $l \in\Bbb N$ and $2\leq n$, the sums of positive powers of eigenvalues of the polyharmonic Laplacian problem (\ref{PL})  on $\Omega$ satisfy
\begin{equation}\label{V1}
  \sum_j^k\ld_j^q\geq \frac{n}{n+2lq}(4\pi)^{\frac{lq}{n}}\left(\frac{\Gamma(1+\frac n2)}{|\Omega|}\right)^{\frac{2lq}{n}}k^{1+\frac{2lq}{n}}.
\end{equation}
\end{cor}

\begin{thm}\label{Vou2}
For $l \in\Bbb N$, $2\leq n$ and $0<p<\frac{n}{2l}$,
 the sums of negative powers of eigenvalues of the
polyharmonic Laplacian problem (\ref{PL})  on $\Omega$ satisfy
\begin{equation}\label{V2}
  \sum_j^k\ld_j^{-p}\leq (4\pi)^{-\frac{lq}{n}}\frac{n}{n-2lq}\left(\frac{|\Omega|}{\Gamma(1+\frac n2)}\right)^{\frac{lq}{n}}k^{1-\frac{2lq}{n}}.
\end{equation}
\end{thm}
\begin{rem}
The coefficients of $k^{1+\frac{2lq}{n}}$ (rep. $k^{1+\frac{2lq-2}{n}}$)  appeared in the inequalities \eqref{V0} and \eqref{V1}(rep. \eqref{V2}) are the best possible from the Weyl's asymptotic formula. Moreover, Corollary \ref{Vou1} and Theorem \ref{Vou2} generalize  the eigenvalue estimates of the Dirichlet Laplacian problem ($l=1$ in \eqref{PL}) in \cite{Vougalter} and  the clamped plate problem ($l=2$ in \eqref{PL}) \cite{YY2013} repectively.
\end{rem}

Another work of this paper is to consider the following eigenvalue problem defined by
\begin{equation}\label{HStokes}
  \left \{\aligned
  &(-\Delta)^l\u_k+\n p_k=\mu_k \u_k, \qquad \text{in}\ \Omega,\\
  &\div \u_k=0,  \qquad \text{in} \   \Omega,\\
  &\u_k|_{\partial \Omega}=\cdots=\frac{\partial^{l-1}\u_k}{\partial \nu^{l-1}}\Big{|}_{\partial \Omega}=0,
  \endaligned \right.
\end{equation}
where $l\in \Bbb N$ and $\Omega$ is a bounded domain with smooth boundary in $\Bbb R^n$.

For $l=1$, (\ref{HStokes}) is the eigenvalue problem of the classical Stokes operator.
Li-Yau type lower bounds for the eigenvalues of the classical Stokes operator were obtained in~\cite{I09}:
\begin{equation}\label{I1}
\sum_{k=1}^k\mu_k\geq 4\pi\frac n{n+2}\left(
\frac{\Gamma(1+\frac n2)}{(n-1)|\Omega|}
\right)^{2/n}k^{1+2/n}.
\end{equation}
The  coefficient of $k^{1+2/n}$  in \eqref{I1} is sharp in view of the asymptotic formula
(cf. \cite{Babenko} when $n=3$ and  \cite{Metiv} when $n\ge2$ )
\begin{equation}\label{I3}
\mu_k\sim
\left(\frac{(2\pi)^n}{\omega_n(n-1)|\Omega|}
\right)^{2/n}k^{2/n},  \quad\text{as}\quad k\to\infty.
\end{equation}
Based on Melas's approach in \cite{M}, Ilyin \cite{I10} proved
\begin{equation}\label{I4}
   \sum_{k=1}^k\mu_k \geq 4\pi\frac n{n+2}\left(
\frac{\Gamma(1+\frac n2)}{(n-1)|\Omega|}
\right)^{2/n}k^{1+2/n}+\frac {(n-1)}{48}\frac {|\Omega|}{I(\Omega)}
\,k\,(1-\varepsilon_n(k)),
\end{equation}
where $0\leq \varepsilon_n(k)=O(k^{-2/n})$.
In 2012, Y. Yolcu and T. Yolcu \cite{YY2012} obtained
\begin{equation}\label{I4}
\begin{aligned}
\sum_{k=1}^k\mu_k \geq & 4\pi\frac n{n+2}\left(
\frac{\Gamma(1+\frac n2)}{(n-1)|\Omega|}
\right)^{2/n}k^{1+2/n}+\frac{2\sqrt{\pi}}{n+2}\left(\frac{n-1}{n}\frac{|\Omega|}{I(\Omega)}  \right)^{\frac12}\left(\frac{\Gamma(1+\frac n2)}{(n-1)|\Omega|} \right)^{\frac 1n}k^{1+\frac1n}\\
&-\frac{5(n-1)}{8n(n+2)}\frac{|\Omega|}{I(\Omega)}k+\frac{1}{16\sqrt{\pi}(n+2)}\left(\frac{n-1}{n}\frac{|\Omega|}{I(\Omega)}  \right)^{\frac32}\left(\frac{\Gamma(1+\frac n2)}{(n-1)|\Omega|} \right)^{-\frac 1n}k^{1-\frac1n}.
\end{aligned}
\end{equation}

For the eigenvalue problem  (\ref{HStokes}), Ilyin \cite{I13} proved
\begin{equation}\label{I5}
\begin{aligned}
\sum_{k=1}^k\mu_k \geq & (4\pi)^l\frac n{n+2}\left(
\frac{\Gamma(1+\frac n2)}{(n-1)|\Omega|}
\right)^{2/n}k^{1+\frac{2l}{n}}
\\&+(4\pi)^{l-1}\frac{l}{48}\frac{(n-1)|\Omega|}
{I(\Omega)}\left(\frac{\Gamma(1+\frac n2)}{(n-1)|\Omega|}\right)^{\frac {2l-2}{n}}k^{1+\frac{2-2l}{n}}(1+\varepsilon_n(k)),\\
\end{aligned}
\end{equation}
where $0  \leq \varepsilon_n(k) = O( k^{-\frac{2}{n}}) $ is a infinitesimal of $k^{-\frac{2}{n}}$.

\begin{thm}\label{HSThm1}For any bounded domain $\Omega\subset \Bbb R^n (n\geq 2), k\geq 1$ and $1\leq l< \frac {n+1}{2}$, the eigenvalues of (\ref{HStokes}) satisfy
\begin{equation}\label{S}
\begin{aligned}
      \sum_{j=1}^k\mu_j\geq & \frac{n(4\pi)^l}{n+2l}\left(\frac{\Gamma(1+\frac n2)}{(n-1)|\Omega|}\right)^{\frac{2l}{n}}k^{1+\frac{2l}{n}}
      +\frac{(4\pi)^{\frac{2l-1}{2}}l}{(n+2l)}
      \left(\frac{\Gamma(1+\frac n2)}{(n-1)|\Omega|}\right)^{\frac{2l-1}{n}}\left(\frac{(n-1)|\Omega|}{nI(\Omega)}\right)^{\frac12}k^{1+\frac{2l-1}{n}}\\
      &-\frac{5(n-1)(4\pi)^{l-1}l}{8n(n+2l)}\left(\frac{\Gamma(1+\frac n2)}{(n-1)|\Omega|}\right)^{\frac{2l-2}{n}}\frac{|\Omega|}{I(\Omega)}k^{1+\frac{2l-2}{n}}\\
      &+\frac{(4\pi)^{\frac{2l-3}{2}}l}{8(n+2l)}\left(\frac{\Gamma(1+\frac n2)}{(n-1)|\Omega|}\right)^{\frac{2l-3}{n}}\left(\frac{(n-1)|\Omega|}{nI(\Omega)}\right)^{\frac32}k^{1+\frac{2l-3}{n}}.
    \end{aligned}
\end{equation}
\end{thm}
\begin{rem}
 Comparing \eqref{S} with \eqref{I5}, one can find that there exists extra term containing $k^{1+\frac{2l-1}{n}}$ in \eqref{S}. Therefore, the inequality \eqref{S} is more precise than the estimate \eqref{I5} under the assumptions of Theorem \ref{HSThm1}. Moreover, it is easy to find that \eqref{S} is a generalization form of \eqref{I4} for problem \eqref{HStokes}.
\end{rem}

\section{Preliminaries}
\subsection{Polyharmonic operator}

Assume that $\Omega\subset \Bbb R^n$ ($n\geq 2$) is a bounded domain in $\Bbb R^n$.
Let $I(\Omega) = \underset{a \in \mathbb{R}^n}{\textrm{min}}  \int_\Omega |x-a|^2 dx$ is the moment of inertia of $\Omega$ and $a$ is a constant vector in $\Bbb R^n$. By translating the open set $\Omega$ if necessary, we may assume that the moment of inertia $I(\Omega)$ is given by
\begin{equation*}
  I(\Omega)=\int_{\Omega}|x|^2dx.
\end{equation*}
 Define

 \begin{eqnarray}
  M=\frac{|\Omega|}{(2\pi)^{n}},  & &L=\frac{2\sqrt{|\Omega|I(\Omega)}}{(2\pi)^{n}},\label{ML} \\
M_\mathrm{S}=\frac{(n-1)|\Omega|}{(2\pi)^{n}},& &
    L_\mathrm{S}=\frac{2(n(n-1))^{1/2}\sqrt{|\Omega|I(\Omega)}} {(2\pi)^{n}}.\nonumber
 \end{eqnarray}
Let $u_j$ be an orthonormal eigenfuction corresponding to the $j$-th eigenvalue $\lambda_j$ of problem (\ref{PL}).
Denote by $\widehat{u}_{j}(\xi)$ the Fourier transform of  $u_{j}(x)$, which is defined by
\begin{equation}
\hat{u}_{j}(\xi)= (2\pi)^{-\frac{n}{2}} \int_{\Omega}  u_{j}(x) e^{ i x \cdot \xi}  dx.
\end{equation}
Set
\begin{equation}\label{F}
 F(\xi) = \sum_{j=1}^k  |\hat{u}_{j}(\xi)|^2.
\end{equation}

From \cite{CSWZ13,I13}, we have
\begin{lem}\label{lem1} The function $F(\xi)$ defined by (\ref{F}) satisfies
  \begin{equation}\label{B1}
    F(\xi)\leq M,
  \end{equation}
  \begin{equation}\label{B2}
    \int_{\Bbb R^n}F(\xi)d\xi=k,
  \end{equation}
  \begin{equation}\label{B3}
    \sum_{j=1}^k|\n \hat u_j(\xi)|^2\leq \frac{I(\Omega)}{(2\pi)^{n}},
  \end{equation}
  \begin{equation}\label{B4}
    |\n F(\xi)|\leq L,
  \end{equation}
 \begin{equation}\label{B5}
   \frac n{n+2}\left(\frac{|\Omega|}{\omega_n}\right)^{\frac 2n}\leq \frac{I(\Omega)}{|\Omega|}
 \end{equation}
 and
\begin{equation}\label{B6}
  \frac{|\Omega|^{1+\frac1n}}{(2\pi)^{n}\omega_n^{\frac1n}}\leq L.
\end{equation}
\end{lem}

Let $F^*(\xi) $ denote the decreasing radial rearrangement of $F(\xi)$. Therefore, by approximating $F(\xi)$, we may assume that there exists a real valued absolutely continuous function $\phi: [0,\infty)\longrightarrow [0,M]$ such that $F^*(\xi)=\phi(|\xi|)$. From \cite{CSWZ13,I13}, we have
\begin{lem}\label{lem2}The function $\phi(s)$  satisfies
  \begin{equation}\label{}
    n\omega_n\int_0^\infty t^{n-1}\phi(t)dt=k,
  \end{equation}
  \begin{equation}\label{}
    \sum_{j=1}^k\ld_j\geq n\omega_n \int_0^\infty t^{n+2l-1}\phi(t)dt
  \end{equation}
and
  \begin{equation}\label{}
    0\leq -\phi(t)\leq L.
  \end{equation}
\end{lem}

\subsection{Stokes operator with higher order}

We firstly recall the functional definition of
the Stokes operator \cite{CF88,Lad,TNS} and its generalization \cite{I13}.
Let $\mathcal{V}$ denote the set of smooth
divergence-free vector functions with compact supports
\begin{equation*}
\mathcal{V}=\{\u:\Omega\to\mathbb{R}^n,
\ \u\in \mathbf{C}^\infty_0(\Omega),\
\div \u=0\}.
\end{equation*}
Let $L$ and $V$ denote the the closures of
$\mathcal{V}$ in  $\mathbf{L}_2(\Omega)$ and
$\mathbf{H}^l_0(\Omega)$ ($1\leq l\in \Bbb N$) respectively.
Moreover, we note that $\mathbf{L}_2(\Omega)$ can be written as $\mathbf{L}_2(\Omega)=L\oplus L^\perp$ (see for instance, \cite{TNS}), where
\begin{equation}\label{}
 \begin{aligned}
L=&\{\u\in \mathbf{L}_2(\Omega)|\div\u=0, \u\cdot \nu|_{\partial \Omega}=0 \},\\
 L^\perp=&\{\u\in \mathbf{L}_2(\Omega)|\u=\n p, p\in L_2^{loc}(\Omega)\}.
 \end{aligned}
\end{equation}
Define the operator $A: V\longrightarrow V'$ by
\begin{equation*}
  {(A\u,\v)}_{V'\times V}:=((-\Delta)^{\frac l2}\u,(-\Delta)^{\frac l2}\v)
\end{equation*}
for $\u,\v \in V$.
The operator $A$ is an isomorphism between $V$ and $V'$. For a sufficient smooth $\u$, we have
\begin{equation}\label{HS}
  A\u=P(-\Delta)^l\u,
\end{equation}
where $P$ is orthogonal projection mapping $\mathbf{L}_2(\Omega)$ to $L$, i.e.
$P:\mathbf{L}_2(\Omega)\longrightarrow L$.
For $l=1$, (\ref{HS}) corresponds to the classical Stokes operator.
The operator $A$ is a self-adjoint positive definite operator with the following discrete spectrum
\begin{equation*}
  A\u_j=\mu_j\u_j,\qquad 0<\mu_1\leq \mu_2\cdots,
\end{equation*}
where $\{\u_j\}_{j=1}^\infty\in V$ are the orthonormal vector eigenfunctions and
\begin{equation*}
  \mu_j=\|(-\Delta)^{\frac{l}{2}}\u_j\|^2.
\end{equation*}

For the orthonormal family
$\{\u_k\}_{k=1}^k\in L$, we set
\begin{equation}\label{Fstokes}
F_\mathrm{S}(\xi)=\sum_{k=1}^m|\hat{\u}_k(\xi)|^2,
\end{equation}
where $\hat{\u}_k(\xi)$ is the Fourier transform of  $\u_{j}(x)$  defined by
\begin{equation}
\hat{\u}_{j}(\xi)= (2\pi)^{-\frac{n}{2}} \int_{\Omega}  \u_{j}(x) e^{ i x \cdot \xi}  dx.
\end{equation}

\begin{lem}The function $F_\mathrm{S}(\xi)$ defined  by (\ref{Fstokes}) satisfies
\begin{equation}\label{enumSt}
\aligned
&0\le F_\mathrm{S}(\xi)\leq
    M_\mathrm{S},\\
&|\nabla F_\mathrm{S}(\xi)|\leq
    L_\mathrm{S},\\
&\int F_\mathrm{S}(\xi)\,d\xi=k,\\
&\int |\xi|^{2l}F_\mathrm{S}(\xi)\,d\xi=
            \sum_{k=1}^k\|(-\Delta)^{\frac{l}{2}}\u_j\|^2=\sum_{k=1}^k\mu_j.
\endaligned
\end{equation}
\end{lem}

Supposing that $F^*_\mathrm{S}(\xi)$ denotes the decreasing radial rearrangement of $F_\mathrm{S}(\xi)$, by approximating
$F_\mathrm{S}(\xi)$, we may infer that there exists a real valued absolutely continuous function
$\phi_\mathrm{S}:[0,\infty)\longrightarrow [0,M_\mathrm{S}]$ such that $F_\mathrm{S}(\xi)=\phi_\mathrm{S}(|\xi|)$. Form \cite{I13}, we have
\begin{lem}\label{HSLemma}The function $\phi_\mathrm{S}(s)$  satisfies
\begin{equation}
\begin{aligned}
  &n\omega_n\int_0^\infty t^{n-1}\phi_\mathrm{S}(t)dt=k,\\
 &n\omega_n \int_0^\infty t^{n+2l-1}\phi_\mathrm{S}(t)dt\leq \sum_{j=1}^k\mu_j,\\
  &0\leq -\phi_\mathrm{S}(t)\leq L_\mathrm{S}.
\end{aligned}
\end{equation}
\end{lem}

\section{Some Lemmas}
In order to prove our main results, we first establish the following lemmas which are motivated by Melas' work in \cite{M}.

\begin{lem}\label{P}
  For $l\in \Bbb N, 2\leq n,$  two positive real numbers $s$ and $t$, we have the following inequalities
  \begin{equation}\label{Polynomial}
    nt^{n+2l}-(n+2l)t^ns^{2l}+2ls^{n+2l}\geq \left(2l s^{n+2l-2}+4lts^{n+2l-3}\right)(s-t)^2.
  \end{equation}
\end{lem}
\begin{proof}
By some direct calculation, for $y>0$, we obtain
  \begin{equation*}
  \begin{aligned}
  ny^{n+2l}-(n+2l)y^n+2l=&(y-1)^2\left(n\sum_{j=2}^{2l}(j-1)y^{n+2l-j}+2l\sum_{j=1}^n(n-j+1)y^{n-j}\right)\\
=&(y-1)^2\left(n\sum_{j=2}^{2l}(j-1)y^{n+2l-j}+2l\sum_{j=1}^{n-2}(n-j+1)y^{n-j}\right)\\
&+(2l+4ly)(y-1)^2.
    \end{aligned}
\end{equation*}
Therefore, we have
\begin{equation*}
  \begin{aligned}
  &ny^{n+2l}-(n+2l)y^n+2l-(2l+4ly)(y-1)^2\\
  =&(y-1)^2\left(n\sum_{j=2}^{2l}(j-1)y^{n+2l-j}
  +2l\sum_{j=1}^{n-2}(n-j+1)y^{n-j}\right)  \geq 0.
      \end{aligned}
\end{equation*}
Setting $y=\frac ts$ in above formula, we get (\ref{Polynomial}).
\end{proof}

\begin{lem}\label{D}
  Suppose that $\beta:[0,\infty)\longrightarrow [0,1] $ such that
  \begin{equation}\label{}
    \int_0^\infty \beta(t)dt=1, \int_0^\infty t^n \beta(t)dt<\infty, \int_0^\infty t^{n+2l}\beta(t)dt<\infty.
  \end{equation}
  Then, there exists $\delta\geq 0$ such that
  \begin{equation}\label{Del1}
    \int_\delta^{\delta+1}t^ndt=\int_0^\infty t^n \beta(t)dt
  \end{equation}
  and
  \begin{equation}\label{Del}
    \int_\delta^{\delta+1}t^{n+2l}dt\leq \int_0^\infty t^{n+2l}\beta(t)dt.
  \end{equation}
\end{lem}
\begin{proof}Note that
 \begin{equation}\label{Chi1}
   (t^n-1)(\beta(t)-\chi_{[0,1]}(t))\geq 0,  \qquad  \text{on} \quad [0,\infty).
 \end{equation}
 Integrating (\ref{Chi1}) over $[0,\infty)$, it gives
 \begin{equation}\label{}
   \int_0^\infty t^n\beta(t)dt \geq \int_0^1t^n=\frac{1}{n+1}.
 \end{equation}
 Therefore, there exists $\delta\geq 0$ such that (\ref{Del1}) holds.

According to Cramer's rule, we can find two positive number $a_1$ and$ a_2$ such that the function
 \begin{equation}\label{}
   q(t)=t^{n+2l}-a_1t^n+a_2
 \end{equation}
 satisfies $q(\delta)=q(\delta+1)=0$. Since $q'(s)$ has at most one zero in $[0,\infty)$, we conclude that \begin{equation*}
  q(s) \left\{\aligned  <0,& \qquad \text{in}\ (\delta,\delta+1),\\
>0,&\qquad \text{in}\ [0,\infty)/ [\delta,\delta+1].
\endaligned\right.
 \end{equation*}
The assumptions on $\beta(t)$ imply that
 \begin{equation}\label{Chi2}
  q(t)(\chi_{[\delta,\delta+1]}(t)-\beta(t))\leq 0, \qquad \text{on} \  [0,\infty).
 \end{equation}
 Integrating the inequality (\ref{Chi2}), taking into account the choice of $\delta$ and using (\ref{Chi1}), we have (\ref{Del}).

\end{proof}

\begin{lem}\label{MainLemma}
Let $n\geq 2, D, A>0$ be contants and $\psi:[0,\infty)\longrightarrow [0,\infty)$ be a decreasing and absolutely continuous function such that
  \begin{equation}\label{C1}
    -D\leq \psi'(t)\leq 0
  \end{equation}
and
  \begin{equation}\label{C2}
    \int_0^\infty t^{n-1}\psi(t)dt=A.
  \end{equation}
Then
  \begin{equation}\label{Mineq2}
  \begin{aligned}
   \int_0^\infty t^{n+2l-1}\psi(t)dt\geq& \frac{1}{n+2l}(nA)^{1+\frac{2l}{n}}\psi(0)^{-\frac{2l}{n}}+
   \frac{2l\varepsilon}{n(n+2l)}(nA)^{1+\frac{2l-1}{n}}\psi(0)^{1-\frac{2l-1}{n}}D^{-1}\\
   &-\frac{5l}{2n(n+2l)}(nA)^{1+\frac{2l-2}{n}}\psi(0)^{2-\frac{2l-2}{n}}D^{-2}
   +\frac{\varepsilon l}{n(n+2l)}(nA)^{1+\frac{2l-3}{n}}\psi(0)^{3-\frac{2l-3}{n}}D^{-3},
  \end{aligned}
  \end{equation}
where $\tau\in (0, 1]$ and $\varepsilon\in (0,1]$.
\end{lem}

\begin{proof} We assume that $B=\int_0^\infty t^{n+2l-1}\psi(t)dt<\infty$. Otherwise there is nothing to prove.
 Define
 \begin{equation}\label{M1}
   S(t)=\frac{1}{\psi(0)}\psi\left(\frac{\psi(0)}{L}t\right) \qquad \text {and }\qquad h(t)=-S'(t),\qquad t\geq 0.
 \end{equation}
Note that $S(0)=1$, $0\leq h(t)\leq 1$ and $\int_0^\infty h(t)dt=S(0)=1$.  Now we define
\begin{equation}\label{M2}
\tilde A=\int_0^\infty t^{n-1}S(t)dt
\end{equation}
and
\begin{equation}
\tilde B=\int_0^\infty t^{n+2l-1}S(t)dt.
\end{equation}

Using the similar arguments as in \cite{YY2012}, we can infer
\begin{equation*}
  \lim \inf_{t\longrightarrow \infty} t^{n+2l}S(t)=0.
\end{equation*}
Moreover, it is not difficult to observe that $\lim \textstyle\inf_{t\longrightarrow \infty} t^{n}S(t)=0 $ as well.
Using the integration by parts, we have
\begin{equation*}
\int_0^\infty t^{n}h(t)dt=n\tilde A
\end{equation*}
and
\begin{equation*}
\int_0^\infty t^{n+2l}h(t)dt=(n+2l)\tilde B.
\end{equation*}
Since $h(t)$ satisfy the conditions in Lemma \ref{D}, there exists $\delta\geq 0$ such that
 \begin{equation}\label{M3}
   \int_\delta^{\delta+1}t^n dt=n\tilde A
 \end{equation}
 and
 \begin{equation}\label{M4}
   \int_\delta^{\delta+1}t^{n+2l} dt\leq \int_0^\infty t^{n+2l}h(t)dt=  (n+2l)\tilde B.
 \end{equation}
Making use of the Jensen's inequality, we derive
 \begin{equation*}
   \frac1{2^n} \leq n\tilde A.
 \end{equation*}
By Lemma \ref{P}, integrating (\ref{Polynomial}) in $t$ from $\delta$ to $\delta+1$, we obtain
\begin{equation}\label{M5}
 \begin{aligned}
   n\int_\delta^{\delta+1}t^{n+2l}dt\geq& (n+2l)s^{2l}\int_\delta^{\delta+1}t^ndt-2ls^{n+2l}+2l s^{n+2l-2}\int_\delta^{\delta+1}(t-s)^2 dt
 \\&+4ls^{n+2l-3}\int_\delta^{\delta+1}t(t-s)^2 dt.
 \end{aligned}
 \end{equation}
Since
\begin{equation*}
  \int_\delta^{\delta+1}(t-s)^2 dt\geq \frac1{12}
\end{equation*}
and
\begin{equation*}
  \int_\delta^{\delta+1}t(t-s)^2 dt\geq\frac 12s^2-\frac23s+\frac14,
\end{equation*}
setting  $s=(n\tilde A)^{\frac{1}{n}}$, we obtain
\begin{equation}\label{}
 n(n+2l)\tilde B\geq n(n\tilde A)^{1+\frac{2l}{n}}+2l(n\tilde A)^{1+\frac{2l-1}{n}}-\frac 52l(n\tilde A)^{1+\frac{2l-2}{n}}+l(n\tilde A)^{1+\frac{2l-3}{n}}
\end{equation}
by using (\ref{M3}) and (\ref{M4}).
Finally, by the definitions of  (\ref{M1}) and (\ref{M2}), we get (\ref{Mineq2}).
\end{proof}

\section{Proofs of  Theorem \ref{PLThm1} and Theorem \ref{HSThm1}}

\begin{proof}[Proof of Theorem \ref{PLThm1}]
From Lemma \ref{lem2} and Lemma \ref{MainLemma}, we have
\begin{equation}\label{Mineq4}
  \begin{aligned}
   \sum_{j=1}^k\ld_j\geq& \frac{n}{n+2l}\omega_n^{-\frac {2l}{n}}\psi(0)^{-\frac{2l}{n}}k^{1+\frac{2l}{n}}+
   \frac{2l\varepsilon}{(n+2l)L}\omega_n^{-\frac {2l-1}{n}}\psi(0)^{1-\frac{2l-1}{n}}k^{1+\frac{2l-1}{n}}\\
   &-\frac{5l}{2(n+2l)L^2}\omega_n^{-\frac {2l-2}{n}}\psi(0)^{2-\frac{2l-2}{n}}k^{1+\frac{2l-2}{n}}
   +\frac{\varepsilon l}{(n+2l)L^3}\omega_n^{-\frac {2l-3}{n}}\psi(0)^{3-\frac{2l-3}{n}}k^{1+\frac{2l-3}{n}}.
  \end{aligned}
\end{equation}

Now define two functions
\begin{equation}\label{}
  F_1(x)=c_1 x^{-\frac {2l}{n}}+c_2x^{1-\frac{2l-1}{n}}\qquad \text{and}\qquad F_2(x)=c_3 x^{3-\frac {2l-3}{n}}-c_4x^{2-\frac{2l-2}{n}},
\end{equation}
for $x\in \left(0,\frac{|\Omega|}{(2\pi)^n}\right]$,
where the constants
\begin{equation*}
  \begin{aligned}
  &c_1=\frac{n}{n+2l}\omega_n^{-\frac {2l}{n}}k^{1+\frac{2l}{n}},
  &c_2=\frac{2l\varepsilon}{(n+2l)L}\omega_n^{-\frac {2l-1}{n}}k^{1+\frac{2l-1}{n}},\\
  &c_3=\frac{\varepsilon l}{(n+2l)L^3}\omega_n^{-\frac {2l-3}{n}}k^{1+\frac{2l-3}{n}},
  &c_4=\frac{5l}{2(n+2l)L^2}\omega_n^{-\frac {2l-2}{n}}k^{1+\frac{2l-2}{n}}.
  \end{aligned}
\end{equation*}
For $1\leq l< \frac{n+1}{2}$, $F_1(x)$ is  decreasing if $0<x\leq \left(\frac{l c_1}{(n-2l+1) c_2} \right)^{\frac{n}{n+1}}$ and  $F_2(x)$ is  decreasing if $ x\leq \left(\frac{(2n+2-2l)c_4}{(3n+3-2l)c_3}\right)^{\frac{n}{n+1}}$. The function $F_1(x)+F_2(x)$ is decreasing for $x\in \left(0,\frac{|\Omega|}{(2\pi)^n}\right]$  when we have
\begin{equation}\label{}
 \frac{|\Omega|}{(2\pi)^n}\leq \min\left\{\left(\frac{l c_1}{(n-2l+1) c_2} \right)^{\frac{n}{n+1}},
 \quad  \left(\frac{(2n+2-2l)c_4}{(3n+3-2l)c_3}\right)^{\frac{n}{n+1}}\right\}
 \end{equation}
Since $k\geq 1$, by (\ref{B6}), it is sufficient to find the upper bound of $\varepsilon$. In fact, we have
\begin{equation}\label{}
  \varepsilon \leq \underset{1\leq l< \frac{n+1}{2},2\leq n}\min\left\{\frac{n}{n-2l+1}  \Gamma(1+\frac n2)^{\frac 2n}, \frac{10(n+1-l)}{3n+3-2l} \Gamma(1+\frac n2)^{\frac{2}{n}}\right\}.
\end{equation}
Obviously, it holds that $\frac{n}{n-1}\leq \frac{n}{n-2l+1}$ and $\frac14<\frac{n+1-l}{3n+3-2l}$ for $1\leq l< \frac{n+1}{2}$
and then $\frac{n}{n-1}<\frac 52$ for $n\geq 2$.
From the definition of $\varepsilon$ in Lemma \ref{MainLemma},  we can choose
 $$
 \varepsilon=\min\left\{1, \frac{n}{n-1}  \Gamma(1+\frac n2)^{\frac 2n}\right\}=1
 $$
to guarantee the function $F_1(x)+F_2(x)$  decreasing for $x\in \left(0,\frac{|\Omega|}{(2\pi)^n}\right]$. Therefore, taking $\varepsilon=1$ and  $\psi(0)=\frac{|\Omega|}{(2\pi)^n}$ in  (\ref{Mineq4}), we can infer (\ref{Mainineq2}).
\end{proof}

\begin{proof}[Proof of Theorem \ref{HSThm1}]
From Lemma \ref{MainLemma}, we have the inequality
\begin{equation*}
  \begin{aligned}
   \sum_{j=1}^k\mu_j\geq& \frac{n}{n+2l}\omega_n^{-\frac {2l}{n}}\psi(0)^{-\frac{2l}{n}}k^{1+\frac{2l}{n}}+
   \frac{2l\varepsilon}{(n+2l)L_\mathrm{S}}\omega_n^{-\frac {2l-1}{n}}\psi(0)^{1-\frac{2l-1}{n}}k^{1+\frac{2l-1}{n}}\\
   &-\frac{5l}{2(n+2l)L_\mathrm{S}^2}\omega_n^{-\frac {2l-2}{n}}\psi(0)^{2-\frac{2l-2}{n}}k^{1+\frac{2l-2}{n}}
   +\frac{\varepsilon l}{(n+2l)L_\mathrm{S}^3}\omega_n^{-\frac {2l-3}{n}}\psi(0)^{3-\frac{2l-3}{n}}k^{1+\frac{2l-3}{n}}.
  \end{aligned}
\end{equation*}
Similar to the proof of Theorem \ref{PLThm1}, we can give the proof of Theorem \ref{HSThm1} by only replacing $\ld_j, L, M$ by $\mu_j, L_\mathrm{S}, M_\mathrm{S}$ in the proof of Theorem \ref{PLThm1} respectively. So we omit its proof.
\end{proof}

\section{Proofs of Theorem \ref{Vou0} and Theorem  \ref{Vou2}  }
In order to prove Theorem \ref{Vou0}, we need the following lemma derived by Ilyin \cite{I10,I13}.

\begin{lem}\label{ILY1}{(Lemma 3.1 in \cite{I10})}\qquad
Let
\begin{equation}\label{Ily1}
\Psi_s(r)  =
\left \{\aligned
&M,            &  \mbox{for}  \quad  0  \leq r \leq s,  \\
&M -  L(r-s),  &  \mbox{for}   \quad  s\leq r \leq s + \frac{M}{L},\\
&0,            &  \mbox{for}  \quad   r \geq s + \frac{M}{L},
\endaligned\right.
\end{equation}
where $M, L$ is given by \eqref{ML}. Suppose that $\int_0^{+\infty}  r^b \Psi_s(r)  dr = m^*$ and $d \geq b$.
Then for any decreasing and absolutely continuous function $F$ satisfying the conditions
\begin{equation}
 0 \leq F \leq M,   \quad  \int_0^{+\infty}   r^b   F(r) dr  =  m^*,   \quad  0 \leq   -F'  \leq L,
 \end{equation}
the following inequality holds:
\begin{equation}
\int_0^{+\infty}   r^d   F(r) dr  \geq   \int_0^{+\infty}   r^d  \Psi_s(r) dr.
\end{equation}
\end{lem}

Now we give the proof Theorem \ref{Vou0}.
\begin{proof}[Proof of Theorem \ref{Vou0}]
Assume that $F(\xi)$ is defined as in (\ref{F}). Define
\begin{equation}\label{V3}
    R=\left(\frac{n+2lq}{n}\frac{(2\pi)^{n}\sum_{j=1}^k\ld_j^q}{|\Omega|\omega_n}\right)^{\frac 1{n+2lq}}
\end{equation}
and
\begin{equation}
    \textit{Y}(\xi)=\frac{|\Omega|}{(2\pi)^n}\chi_{B_R(0)}(\xi),
\end{equation}
such that
\begin{equation}\label{}
  \int_{\Bbb R^n}|\xi|^{2lq}\textit{Y}(\xi)d\xi=\sum_{j=1}^k\ld_j^q.
\end{equation}
Since $\int_{\Bbb R^n} |\hat u_j(\xi)|^2d\xi=1$ and $x\longmapsto x^q$ is concave for $x\geq 0$ and $q\in (0, 1]$,
by using the Jesen's inequality, we can derive
\begin{equation}\label{IV1}
\sum_{j=1}^k\ld_j^q=\sum_{j=1}^k\left(\int_{\Bbb R^n}|\xi|^{2l} |\hat u_j(\xi)|^2d\xi\right)^q\geq
\int_{\Bbb R^n}|\xi|^{2lq} F(\xi)d\xi.
\end{equation}
Let $F^*(\xi) $ denote the decreasing radial rearrangement of $F(\xi)$.
According to the Hardy-Littlewood inequality, we have
\begin{equation}\label{IV2}
 \int_{\Bbb R^n}|\xi|^{2lq} F(\xi)d\xi\geq   \int_{\Bbb R^n}|\xi|^{2lq} F^*(\xi)d\xi.
\end{equation}
By using \eqref{IV1} and \eqref{IV2}, we can get
\begin{equation}\label{IV3}
\sum_{j=1}^k\ld_j^q\geq \int_{\Bbb R^n}|\xi|^{2lq} F^*(\xi)d\xi=n\omega_n\int_0^\infty r^{n-1+2lq}F^*(r)dr,
\end{equation}
where $F^*(r)=F^*(|\xi|)$.   It is easy to find $F(\xi)$ and $F^*(|\xi|)$ satisfy the conditions of Lemma \ref{ILY1}
according to Lemma \ref{lem1}.
Therefore, we have
\begin{equation}\label{IV4}
\int_0^{+\infty}   r^{n-1+2lq}  F^*(r) dr  \geq   \int_0^{+\infty}  r^{n-1+2lq} \Psi_s(r) dr.
\end{equation}
Using Corollary 1 in \cite{I13}, we find that
\begin{equation}\label{IV5}
 \begin{aligned}
n\omega_n\int_0^{+\infty}  r^{n-1+2lq} \Psi_s(r) dr \geq & \frac{n}{n+2lq}\left(\frac{1}{\omega_nM}\right)^{\frac{2lq}{n}}n\omega_n
   k^{1+\frac{2lq} n}\\
   &+\frac{2lq n}{24}\frac{M^2}{L^2}\left(\frac{1}{\omega_nM}\right)^{\frac {2lq-2}{n}}n\omega_n k^{1+\frac{2lq-2} n}+O(k^{1+\frac{2lq-4}{n}})\\
   =&(4\pi)^{lq}\frac{n}{n+2lq}\left(\frac{\Gamma(1+n/2)}{|\Omega|}\right)^{\frac{2lq}{n}}
   k^{1+\frac{2lq}{n}}\\
   &+(4\pi)^{lq-1}\frac{nlq}{48}\frac{|\Omega|}{I(\Omega)}
   \left(\frac{\Gamma(1+n/2)}{|\Omega|}\right)^{\frac{2lq-2}{n}}k^{1+\frac{2lq-2}{n}}+O(k^{1+\frac{2lq-4}{n}}),
  \end{aligned}
\end{equation}
where  we use the definitions of $M$ and $L$ given by \eqref{ML}.
Finally, from \eqref{IV3}$\sim$\eqref{IV5}, we obtain \eqref{V0}.
\end{proof}

Now we give the proof of Theorem \ref{Vou2}.
\begin{proof}[Proof of Theorem \ref{Vou2}]
Define
\begin{equation}\label{V4}
  \tilde{R}=\left(\frac{(n-2lp)}{n}\frac{(2\pi)^n\sum_{j=1}^k\ld_{j=1}^{-p}}
              {|\Omega|\omega_n}\right)^{\frac{1}{n-2lp}}
\end{equation}
and
\begin{equation}
 \tilde{Y}(\xi)=\frac{|\Omega|}{(2\pi)^n}\chi_{B_{\tilde R(0)}}(\xi)
\end{equation}
such that
\begin{equation}\label{}
  \int_{\Bbb R^n}\frac{\tilde{Y}(\xi)}{|\xi|^{2lp}}d\xi=\sum_{j=1}^k\ld_j^{-p}.
\end{equation}
Since $|\hat u_j(\xi)|^2 d\xi$ defines a probability measure on $\Bbb R^n$ , $x\longmapsto x^{-p}$ is convex for $x>0$ and $p>0$,
we derive
\begin{equation}\label{}
  \int_{\Bbb R^n} \frac{F(\xi)}{|\xi|^{2lp}}\geq \sum_{j=1}^k\ld_j^{-p}
\end{equation}
according to the Jensen's inequality.
Similar to the arguments in Lemma 1 in \cite{LY}, one can infer
\begin{equation*}
  \int_{\Bbb R^n} F(\xi)d\xi\geq \int_{\Bbb R^n} \tilde Y(\xi)d\xi.
\end{equation*}
By using (\ref{B2}) and  (\ref{V4}), we obtain (\ref{V2}).
\end{proof}

\begin{rem}
The proof of Theorem \ref{Vou1} can also be given directly by using a similar approach as the proof of Theorem \ref{Vou2}.
\end{rem}

\bigskip
\noindent {\bf Acknowledgements}
 The work of the first named author was partially supported by NSFC grant No.11101234. The work of the second named author  was partially
supported by NSFC grant No. 11001130.


\providecommand{\bysame}{\leavevmode\hbox
to3em{\hrulefill}\thinspace}

\begin{flushleft}
Daguang Chen\\
Department of Mathematical Sciences, Tsinghua University, Beijing, 100084, P. R. China \\
E-mail: dgchen@math.tsinghua.edu.cn
\end{flushleft}

\begin{flushleft}
He-Jun Sun\\
Department of Applied Mathematics, College of Science, Nanjing University of Science and Technology,
Nanjing 210094, P. R. China \\
E-mail: hejunsun@163.com
\end{flushleft}

\end{document}